\documentclass[12pt]{article}
\usepackage{slashbox}
\usepackage{multirow}
\usepackage[T1]{fontenc}
\usepackage[cp1250]{inputenc}

\usepackage{graphicx}
\usepackage{color}
\usepackage{enumerate}
\usepackage{rotating}
\usepackage{tikz}
\usepackage{amsmath, amsthm, amsfonts, amssymb}

\newtheorem{theorem}{Theorem}[section]

\newtheorem{corollary}[theorem]{Corollary}

\newtheorem{proposition}[theorem]{Proposition}

\newtheorem{lemma}[theorem]{Lemma}

\begin{document}

\markboth{ }{}
\title{\bf Equitable coloring of corona products of cubic graphs is harder than ordinary coloring\footnote{This project has been 
partially supported by Narodowe Centrum Nauki under contract 
DEC-2011/02/A/ST6/00201}}
\date{}
\author{Hanna Furmańczyk\footnote{Institute of Informatics,\ University of Gdańsk,\ Wita Stwosza 57, \ 80-952 Gdańsk, \ Poland. \ 
e-mail: hanna@inf.ug.edu.pl},  \ Marek Kubale \footnote{Department of Algorithms and System Modelling,\ Technical University of Gdańsk,\ Narutowicza 11/12, \ 80-233 Gdańsk, \ Poland. \ 
e-mail: kubale@eti.pg.gda.pl}
}

\markboth{H. Furmańczyk, M. Kubale}{Equitable Colorings of Corona Products of Cubic Graphs ...}

\maketitle

\begin{abstract}
A graph is equitably $k$-colorable if its vertices can be partitioned into $k$ independent 
sets in such a way that the number of vertices in any two sets differ by at most one. 
The smallest $k$ for which such a coloring exists is known as the \emph{equitable chromatic 
number} of $G$ and it is denoted by $\chi _{=}(G)$. In this paper the problem of determinig $\chi_=$ for coronas of cubic graphs is studied. Although the problem of ordinary coloring of coronas of cubic graphs is solvable in polynomial time, the problem of equitable coloring becomes NP-hard for these graphs. We provide polynomially solvable cases of coronas of cubic graphs and prove the NP-hardness in a general case. As a by-product we obtain a simple linear time algorithm for equitable coloring of such graphs which uses $\chi_=(G)$ or $\chi_=(G)+1$ colors. Our algorithm is best possible, unless $P=NP$. Consequently, cubical coronas seem to be the only known class of graphs for which equitable coloring is harder than ordinary coloring.
\end{abstract}


\section{Introduction}
All graphs considered in this paper are connected, finite and simple, i.e. undirected, loopless and without multiple edges, unless otherwise is stated.

If the set of vertices of a graph $G$ can be partitioned into $k$ (possibly empty) classes $V_1,V_2,....,V_k$ such that each $V_i$ is an independent set and the condition 
$\big||V_i|-|V_j|\big|\leq 1$ holds for every pair ($i, j$), then $G$ is said to be {\it equitably k-colorable}. If $|V_i|=l$ for every $i=1,2,\ldots,k$, then $G$ on $n=kl$ vertices is said to be \emph{strong equitably $k$-colorable}.
The smallest integer $k$ for which $G$ is equitably $k$-colorable is known as the {\it equitable chromatic number} of $G$ and it is denoted by $ \chi_{=}(G)$ \cite{meyer}.
Since equitable coloring is a proper coloring with an additional constraint, we have $\chi(G) \leq \chi_=(G)$ for any graph $G$.

The notion of equitable colorability was introduced by Meyer \cite{meyer}. 
However, an earlier work of Hajnal and Szemer\'edi \cite{hfs:haj} showed that 
a graph $G$ with maximal degree $\Delta$ is equitably $k$-colorable 
if $k\geq\Delta+1$. Recently, Kierstead et al. \cite{fast} have given an $O(\Delta n^2)$-time algorithm for obtaining a $(\Delta+1)$-coloring of a graph $G$ on $n$ vertices.

\par This model of graph coloring has many practical applications. Every time when we have to divide a system with binary conflict relations into equal or almost equal conflict-free subsystems we can model this situation by means of equitable graph coloring. 
In particular, one motivation for equitable coloring suggested by Meyer \cite{meyer} concerns scheduling problems. In this application, the vertices of a graph represent a collection of tasks to be performed and an edge connects two tasks that should not be performed at the same time. A coloring of this graph represents a partition of tasks into subsets that may be performed simultaneously. Due to load balancing considerations, it is desirable to perform equal or nearly-equal numbers of tasks in each time slot, and this balancing is exactly what equitable colorings achieve. Furmańczyk \cite{furm} mentions a specific application of this type of scheduling problem, namely, assigning university courses to time slots in a way that avoids scheduling incompatible courses at the same time and spreads the courses evenly among the available time slots.

The topic of equitable coloring was widely discussed in literature. It was considered for some particular graph classes and also for several graph products: cartesian, weak or strong tensor products \cite{prod, furm} as well as for coronas \cite{hf, kaliraj}.

The \emph{corona} of two graphs $G$ and $H$ is the graph $G \circ H$ obtained by taking one copy of $G$, called the \emph{center graph}, $|V(G)|$ copies of $H$, named the 
\emph{outer graph}, and making the $i$-th vertex of $G$ adjacent to every vertex in the $i$-th copy of $H$. Such type of graph products was introduced by Frucht and Harary in 1970 \cite{frucht} (for an example see Fig. \ref{rysex}).

In general, the problem of optimal equitable coloring, in the sense of the number of colors used, is NP-hard and remains so for corona products of graphs. In fact, Furmańczyk et al. \cite{hf} proved that the problem of deciding whether $\chi _{=}(G \circ K_2) \leq 3$ is NP-complete even if $G$ is restricted to the line graph of a cubic graph.

Let us recall some basic facts concerning cubic graphs. It is well known from Brook's theorem \cite{brooks} that for any cubic graph $G \neq K_4$, we have $\chi(G) \leq 3$. 
On the other hand, Chen et al. \cite{clcub} proved that for any cubic graph with $\chi(G)=3$, its equitable chromatic number equals 3 as well. Moreover, since a connected cubic 
graph $G$ with $\chi(G)=2$ is a bipartite graph with partition sets of equal size, we have the equivalence of the classical and equitable chromatic numbers for 2-chromatic 
cubic graphs. Since the only cubic graph for which the chromatic number is equal to 4 is the complete graph $K_4$, we have
\begin{equation}2 \leq \chi_=(G) =\chi(G)\leq 4,\label{brooks}\end{equation}
for any cubic graph $G$.

In the paper we will consider the equitable coloring of coronas. We assume that in corona $G \circ H$, $|V(G)|=n$ and $|V(H)|=m$. A vertex with color $i$ is called an $i$-\emph{vertex}. We use color 4 instead of 0, in all colorings in the paper, including cases when color label is implied by an expresion $(\bmod 4)$. 

Let
\begin{itemize}
\item $Q_2$ denote the class of equitably $2$-chromatic cubic graphs,
\item $Q_3$ denote the class of equitably $3$-chromatic cubic graphs,
\item $Q_4$ denote the class of equitably $4$-chromatic cubic graphs. 

Clearly, $Q_4=\{K_4\}$.
\end{itemize}

Next, let $Q_2(t) \subset Q_2$ ($Q_3(t) \subset Q_3$) denote the class of bipartite (tripartite) cubic graphs with partition sets of cardinality $t$, and let $Q_3(u,v,w) \subset Q_3$ denote the class of 
3-partite graphs with color classes of cardinalities $u$, $v$ and $w$, respectively, where $u \geq v \geq w \geq u-1$. Observe that 
\begin{equation}
\chi(K_4 \circ H) = \left \{
\begin{array}{ll}
4 & \text{if } H \in Q_2,\\  
\chi(H)+1 & \text{otherwise.}
\end{array}\right.
\end{equation}

In the next section we show a way to color $G\circ H$ with 3 colors provided that the corona admits such a coloring. Next, in Section \ref{5eq} we give a linear-time procedure for coloring 
corona products of cubic graphs with 5 colors. It turns out that this number of colors is sufficent for equitable coloring of any corona of cubic graphs, but in some cases less 
than 5 colors suffice. In Section \ref{NP} we give our main result that deciding whether $G \circ H$ is equitably 4-colorable is NP-complete when $H \in Q_3(t)$ and 10 divides $t$, in symbols 
$10 | t$. Hence, our 5-coloring algorithm of Section \ref{5eq} is 1-absolute approximate and the problem of equitable coloring of cubical coronas belongs to very few NP-hard problems 
that have approximation algorithms of this kind. 
Most of our results are summarized in Table~\ref{tabela1}.

\begin{table}[htb]
\begin{center}
\begin{tabular}{|c|*{3}{c|}}\hline

\backslashbox[20mm]{$G$}{$H$} & $Q_2$ & $Q_3$ & $Q_4$\\ \hline
$Q_2$ & 3 or 4 \scriptsize{[Thm. \ref{3-4}]} & 4 or 5$^*$ \scriptsize{[Thms. \ref{2_3_5}, \ref{twNP}]} &5 \scriptsize{[Thm. \ref{k4}]}\\ \hline
$Q_3$ & 3 or 4 \scriptsize{[Thm. \ref{3-4}]} & 4 or 5$^*$ \scriptsize{[Thm. \ref{Cub3_xxx}, Col. \ref{NP2}]} &5 \scriptsize{[Thm. \ref{k4}]}\\ \hline
$Q_4$ & 4 \scriptsize{[Thm. \ref{3-4}]} & 4 & 5 \scriptsize{[Thm. \ref{k4}]}\\\hline
\end{tabular}

\vspace{3mm}
\caption{Possible values of $\chi_=(G \circ H)$, where $G$ and $H$ are cubic graphs. Asterix ($^*$) means that deciding this case is NP-complete.}\label{tabela1}
\end{center}
\end{table}

To the best of our knowledge, cubical coronas are so far the only class of graphs for which equitable coloring is harder than 
ordinary coloring. And, since $\chi_=(G \circ H) \leq 5$ and $\Delta(G \circ 
H) \geq 7$, our results confirm Meyer's Equitable Coloring Conjecture \cite{meyer}, which claims that for any connected graph $G$, other than a complete graph or an odd cycle, 
we have $\chi_=(G) \leq \Delta$.

\section{Equitable 3-coloring of corona of cubic graphs} \label{sec3-4}
First, let us recall a result concerning coronas $G \circ H$, where $H$ is a 2- or 3-partite graph.

\begin{theorem}[\cite{hf}]
Let $G$ be an equitably $k$-colorable graph on $n \geq k$  vertices and let $H$ be a $(k-1)$-partite graph. If $k|n$, then
$$\chi_=(G \circ H) \leq k.$$ \label{bip}
\end{theorem}

\begin{proposition}
If $G$ and $H$ are cubic graphs, then
$\chi_=(G \circ H) =3$ if and only if $G \in Q_2 \cup Q_3$, $H \in Q_2$, and $G$ has a strong equitable 3-coloring.\label{3col}
\end{proposition}
\begin{proof}
\noindent $(\Leftarrow)$ Since $G$ is strong equitably 3-colorable, the cardinality of its vertex set must be divisible by 3. The thesis follows now from Theorem \ref{bip}.

\noindent $(\Rightarrow)$ Assume that $\chi_=(G \circ H) =3$. This implies:
\begin{itemize}
\item $H$ must be 2-chromatic, and due to (\ref{brooks}) it must be also equitably 2-chromatic,
\item $G$ must be 3-colorable (not necessarily equitably), $\chi(G) \leq \chi_=(G)\leq 3$, which implies $G \in Q_2 \cup Q_3$.
\end{itemize}
Otherwise, we would have $\chi(G \circ H)  \geq 4$ which is a contradiction. 

Since $H\in Q_2$ is connected, its bipartition is determined. Let $H\in Q_2(t)$, $t\geq 3$. Observe that every 3-coloring of $G$ determines a 3-partition of 
$G \circ H$. Let us consider any 3-coloring of $G$ with color classes of cardinality $n_1, n_2$ and $n_3$, respecively, where $n=n_1+n_2+n_3$. Then the cardinalities of 
color classes in the implied 3-coloring of $G \circ H$ form a sequence $((n_2+n_3)t, (n_1+n_3)t, (n_1+n_2)t)$. Such a 3-coloring of $G \circ H$ is equitable if and 
only if $n_1=n_2=n_3$. This means that $G$ must have a strong equitable 3-coloring, which, keeping in mind that $\chi_=(G\circ H) \geq 3$ for all cubic graphs $G$ and $H$, completes the proof.
\end{proof}

In the remaining cases of coronas $G \circ H$, where $H \in Q_2$, we have to use more than three colors. However, it turns out that in all such cases four colors suffice.

\begin{theorem}
If $G$ is a cubic graph, $H \in Q_2$, then
$$\chi_=(G \circ H)=
\left \{
\begin{array}{ll}
3 & \text{if } G\in Q_2(s) \cup Q_3, 3 | s \text{ and } G \text{ is equitably 3-colorable},\\
4 & \text{otherwise.}
\end{array}
\right.
$$\label{3-4} 
\end{theorem}
\begin{proof}
Due to Proposition \ref{3col}, we only have to define an equitable 4-coloring of $G \circ H$. The cases of $G \in Q_2 \cup Q_4$ are easy. We start from an equitable 4-coloring of the center graph and extend it to the corona. 

Let us assume that $G \in Q_3$. First, we color equitably $G$ with 3 colors and then extend this coloring to equitable 4-coloring of $G\circ H$, $H=H(U,V)\in Q_2(t)$. 
Since the number of vertices of cubic graph $G$ is even, we have to consider two cases.
\begin{description}
\item[\textnormal{\emph{Case} 1:}] $n=4k$, for some $k\geq 2$.
 
Since $G$ is equitably 3-colorable, the color classes of equitable 3-coloring of $G$ are of cardinalities $\lceil 4k/3\rceil, \lceil (4k-1)/3\rceil$ and $\lceil(4k-2)/3\rceil$, 
respectively. And, since $|V(G \circ H)|=4k(2t+1)$, in every equitable 4-coloring of $G \circ H$ each color class must be of cardinality $2kt+k$. 

We extend our 3-coloring of $G$ to $G \circ H$ as follows (see Fig. \ref{rysex}a)). We color:
\begin{itemize}
\item the vertices in one copy of $H$ linked to a 1-vertex in $G$ using $t$ times color 3 (vertices in partition $U$), 
$t-(\lceil (4k-1)/3\rceil -k)$ times color 2 and $\lceil (4k-1)/3 \rceil -k$ times color 4 (vertices in partition $V$),
\item the vertices in one copy of $H$ linked to a 2-vertex in $G$ using $t$ times color 1 (vertices in partition $U$), 
$t-(\lceil (4k-2)/3 \rceil - k)$ times color 3 and $\lceil (4k-2)/3 \rceil - k$ times color 4 (vertices in partition $V$),
\item  the vertices in one copy of $H$ linked to a 3-vertex in $G$ using $t$ times color 2 (vertices in partition $U$), 
$t-(\lceil 4k/3 \rceil - k)$ times color 1 and $\lceil 4k/3 \rceil - k$ times color 4 (vertices in partition $V$).
\end{itemize}

\begin{figure}[htb]
\begin{center}
\includegraphics[scale=1]{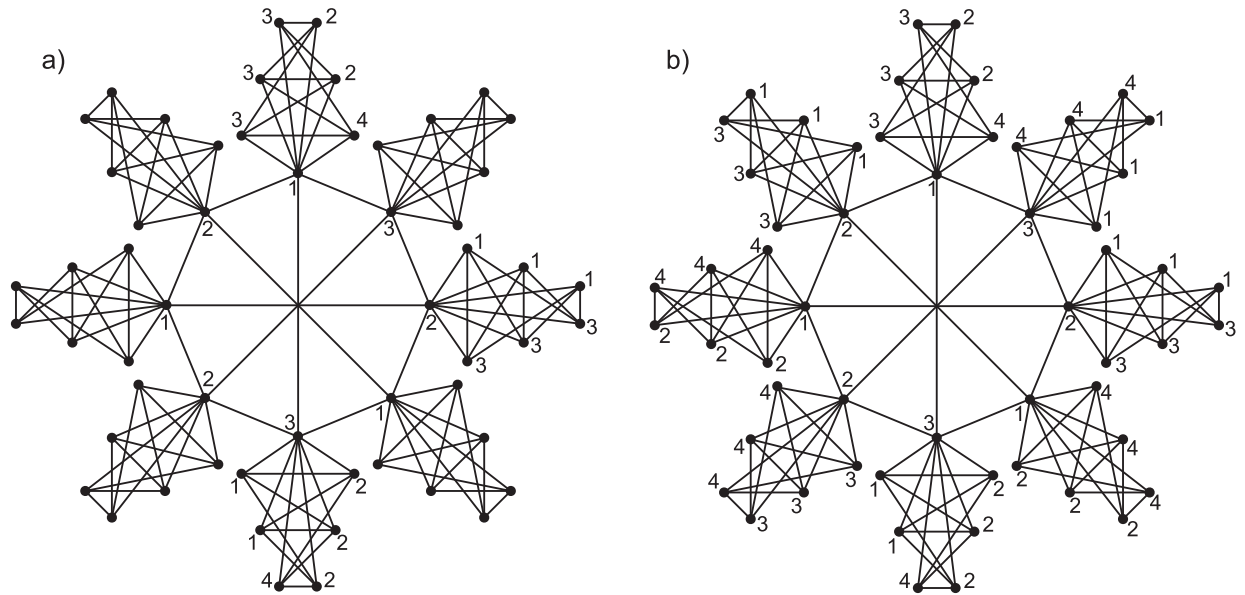} 
\caption{An example of coloring of $W \circ K_{3,3}$, where $W$ is the Wagner graph ($C_8$ with 4 diagonals): a) partial 4-coloring; b) equitable 4-coloring.}
\label{rysex}
\end{center}
\end{figure}

So far, colors 1, 2 and 3 have been used $2t+k$ times, while color 4 has been used $k$ times. 

Now, we color each of uncolored copy of $H$ with two out of three allowed colors in such a way that in this step colors 1, 2 and 3 are used $(2k-2)t$ times and color 4 is 
used $2kt$ times, which results in an equitable 4-coloring of the whole corona $G \circ H$ (see Fig. \ref{rysex}b)). 

\item[\textnormal{\emph{Case} 2:}] $n=4k+2$, for some $k\geq 1$.

Since $G$ is equitably 3-colorable, its color classes are of cardinalities $\lceil (4k+2)/3\rceil$, $\lceil (4k+1)/3\rceil$ and $\lceil 4k/3\rceil$, respectively, in any 
equitable coloring of $G$. Since $|V(G \circ H)|=(4k+2)(2t+1)=8kt+4t+4k+2$, in every equitable 4-coloring the color classes must be of cardinality $2kt+t+k$ or 
$2kt+t+k+1$.

We color:
\begin{itemize}
\item the vertices in one copy of $H$ linked to a 1-vertex of $G$ using $t$ times color 3 (vertices in partition $U$), $t-(\lceil (4k+1)/3\rceil -k-1)$ times color 2 and 
$\lceil (4k+1)/3 \rceil -k-1$ times color 4 (vertices in partition $V$),
\item the vertices in one copy of $H$ linked to a 2-vertex of $G$ using $t$ times color 1 (vertices in partition $U$), $t-(\lceil 4k/3 \rceil - k)$ times color 3 and 
$\lceil 4k/3 \rceil - k$ times color 4 (vertices in partition $V$),
\item  the vertices in one copy of $H$ linked to a 3-vertex of $G$ using $t$ times color 2 (vertices in partition $U$), $t-(\lceil (4k+2)/3 \rceil - k-1)$ times color 1 
and $\lceil (4k+2)/3 \rceil - k-1$ times color 4 (vertices in partition $V$).
\end{itemize}

So far, colors 1 and 2 have been used $2t+k+1$ times, while color 3 has been used $2t+k$ times and color 4 has been used $k$ times. 

Finally, we color still uncolored copies of $H$ with two (out of three) allowed colors so that colors 1, 2 and 3 are used $(2k-1)t$ times and color 4 is 
used $2kt$ times, which results in an equitable 4-colorings\ of the whole corona $G \circ H$. 
\end{description}
\end{proof}
\section{Equitable 5-coloring of coronas of cubic graphs} \label{5eq}

We start by considering cases when 5 colors are necessary for such graphs to be colored equitably.

\begin {proposition} [\cite{hf}]
If $G$ is a graph with $\chi\left(G\right)\leq m+1$, then  $\chi_{=}(G \circ K_m)= m+1$. \label{complete}
\end{proposition}

This proposition immediately implies

\begin{corollary}
If $G$ is a cubic graph, then $$\chi_=(G \circ K_4)=5.$$ \label{k4}
\end{corollary}

It turns out that 5 colors may be required also in some coronas $G \circ H$, where $G \in Q_2 \cup Q_3$ and $H \in Q_3$.

\begin{theorem}
If $G \in Q_2(s)$ and $H \in Q_3$, then $$4\leq \chi_=(G\circ H) \leq 5.$$ \label{2_3_5}
\end{theorem}
\begin{proof}
Since $H \in Q_3$, we obviously have $\chi_=(G \circ H) \geq 4$.

To prove the upper bound, we consider two cases. Let $H =H(U,V,W)$ with tripartition of $H$ satisfying $|U| \geq |V| \geq |W|$.
\begin{description}
\item[\textnormal{\emph{Case} 1:}]  $s = 2k+1$, $k \geq 1$.

We start with the following 4-coloring of $G \circ H$.
\begin{enumerate}
\item Color graph $G$ with 4 colors, using each of colors 1 and 2 $k$ times and colors 3 and 4 ($k+1$) times, respectively.
\item Color the vertices of each copy of $H(U,V,W)$ linked to an $i$-vertex of $G$ using color $(i+1) \bmod 4$ for vertices in $U$, color $(i+2) \bmod 4$ for vertices in $V$, and color $(i+3) \bmod 4$ for vertices in $W$ (we use color 4 instead of 0).
\end{enumerate}

Now, we have to consider three subcases, where we bound the number of vertices that have to be recolored to 5.

\begin{description}
\item[\textnormal{\emph{Subcase} 1.1:}] $H \in Q_3(t+1,t,t)$, where $t=v=w$.

The color sequence of the 4-coloring of this corona is  $\mathcal{C}_4=(c_1,c_2,c_3,c_4)=(3kt+2k+2t+1,3kt+2k+2t,3kt+2k+t+1,3kt+2k+t+2)$.

In every equitable 5-coloring of the corona $G \circ H$, where $G\in Q_2(2k+1)$ and $H \in Q_3(t+1,t,t)$, every color must be used 
$\gamma_5^1=\lceil (12kt+8k+6t+4)/5\rceil=(2kt+t+k+\lceil (2kt+t+3k+4)/5\rceil)$ or $\gamma_5^2=(2kt+t+k+\lfloor (2kt+t+3k+4)/5\rfloor)$ times.
The number $d_i$ of vertices colored with $i$, $1 \leq i \leq 4$, that have to be recolored is equal to $c_i - \gamma_5^1$ or $c_i - \gamma_5^2$. We have
$$d_1\leq c_1 - \gamma_5^1 \leq c_1 - \gamma_5^2 = kt+t+k+1 - \lfloor (2kt+t+3k+4)/5\rfloor =$$ $$=(k+1)(t+1)-\lfloor (2kt+t+3k+4)/5\rfloor\leq (k+1)(t+1).$$
Similarly, we have
\begin{eqnarray*}
d_2 &\leq &k(t+1)+t,\\
d_3 &\leq &k(t+1), \text{ and}\\
d_4 &\leq &k(t+1).
\end{eqnarray*}

\item[\textnormal{\emph{Subcase} 1.2:}] $H \in Q_3(t+1,t+1,t)$, where $t=w$.

The color sequence of the 4-coloring of this corona is $\mathcal{C}_4=(c_1,c_2,c_3,c_4)=(3kt+3k+2t+2,3kt+3k+2t+1,3kt+3k+t+1,3kt+3k+t+2)$.

In every equitable 5-coloring of the corona $G \circ H$, where $G \in Q_2(2k+1)$ and $H\in Q_3(t+1,t+1,t)$, every color must be used 
$\gamma_5^1=\lceil (12kt+12k+6t+6)/5\rceil=(2kt+t+2k+1+\lceil (2kt+t+2k+1)/5\rceil)$ or $\gamma_5^2=(2kt+t+2k+1+\lfloor (2kt+t+2k+1)/5\rfloor)$ times.

Similarly, as in Subcase 1.1, we have
\begin{eqnarray*}
d_1 &\leq &c_1 - \gamma_5^1 \leq c_1 - \gamma_5^2 \leq (k+1)(t+1),\\
d_2 &\leq &k(t+1)+y,\\
d_3 &\leq & k(t+1),\text{ and}\\
d_4 &\leq & k(t+1).
\end{eqnarray*}
\item[\textnormal{\emph{Subcase} 1.3:}] $H \in Q_3(t)$, where $t=u=v=w$.

The color sequence of the 4-coloring of this corona is $\mathcal{C}_4=(c_1,c_2,c_3,c_4)=(3kt+k+2t,3kt+k+2t,3kt+k+t+1,3kt+k+t+1)$.

In every equitable 5-coloring of the corona $G \circ H$, where $G \in Q_2(2k+1)$ and $H \in Q_3(t,t,t)$, every color must be used 
$\lceil (12kt+4k+6t+2)/5\rceil=(2kt+t+\lceil (2kt+t+4k+2)/5\rceil)$ or $(2kt+t+\lfloor (2kt+t+4k+2)/5\rfloor)$ times.

Similarly, as in previous subcases, we have
\begin{eqnarray*}
d_1 &\leq & (k+1)t,\\
d_2 &\leq &kt+t,\\
d_3 &\leq & kt,\text{ and}\\
d_4 &\leq &kt.
\end{eqnarray*}
\end{description}

Consequently, in all subcases, the number of $i$-vertices that have to be recolored is bounded by:
\begin{itemize}
\item $(k+1)u$ for $i=1$,
\item $ku+w$ for $i=2$,
\item $ku$ for $i=3,4$.
\end{itemize}

To obtain an equitable 5-coloring from the 4-coloring of $G \circ H(U,V,W)$, $|U|\geq |V| \geq |W|$, we recolor the appropriate number of $i$-vertices in partitions $U$ linked to $(i-1)$-vertices of $G$ for the vertices which were colored with color $i$. Due to the above, this is possible in the cases of 
colors 1, 3 and 4. In the case of 2-vertices, the number of vertices recolored  in partition $U$ in copies of $H$ can be insufficient. 
In this case, we can recolor the vertices in partition $W$ (of cardinality $w$) in one copy of $H$ linked to 3-vertex of $G$.

\item[\textnormal{\emph{Case} 2:}]  $s = 2k$, $ k\geq 2$.

Again, we start with 4-coloring of $G \circ H$, as follows.
\begin{enumerate}
\item Color graph $G$ with 4 colors, using each of colors 1,2, 3 and 4 $k$ times.
\item Color the vertices of each copy of $H(U,V,W)$ linked to an $i$-vertex of $G$ using color $(i+1) \bmod 4$ for vertices in $U$, color $(i+2) \bmod 4$ for 
vertices in $V$, and color $(i+3) \bmod 4$ for vertices in $W$ (we use color 4 instead of 0).
\end{enumerate}

Notice that the resulting 4-coloring does not require recoloring: it is equitable and establishes that the lower bound is tight.
\end{description}
\end{proof}

Similar technique for obtaining an equitable coloring is used in the proof of the following theorem, by introducing the fifth color.

\begin{theorem}
If $G, H \in Q_3$, then $$4 \leq \chi_=(G\circ H) \leq 5.$$ \label{Cub3_xxx}
\end{theorem}
\begin{proof}
Let $G=G(A,B,C)$, where $|A| \geq |B| \geq |C| \geq |A| - 1$, and let $H=H(U,V,W)$, where $|U|\geq |V| \geq |W| \geq |U|-1$. We start with a 4-coloring of $G \circ H$.
\begin{enumerate}
\item Color the vertices of graph $G$ with 3 colors: the vertices in $A$ with color 1, in $B$ with 2, and in $C$ with color 3.
\item Color the vertices of each copy of $H$ linked to an $i$-vertex using color $(i+1) \bmod 4$ for vertices in $U$, color $(i+2) \bmod 4$ for vertices in $V$, and color $(i+3) \bmod 4$ for vertices in $W$, $i=1,2,3$ (see Fig. \ref{zPrism}a)).
\end{enumerate} 

\begin{figure}[htb]
\begin{center}
\includegraphics[scale=1]{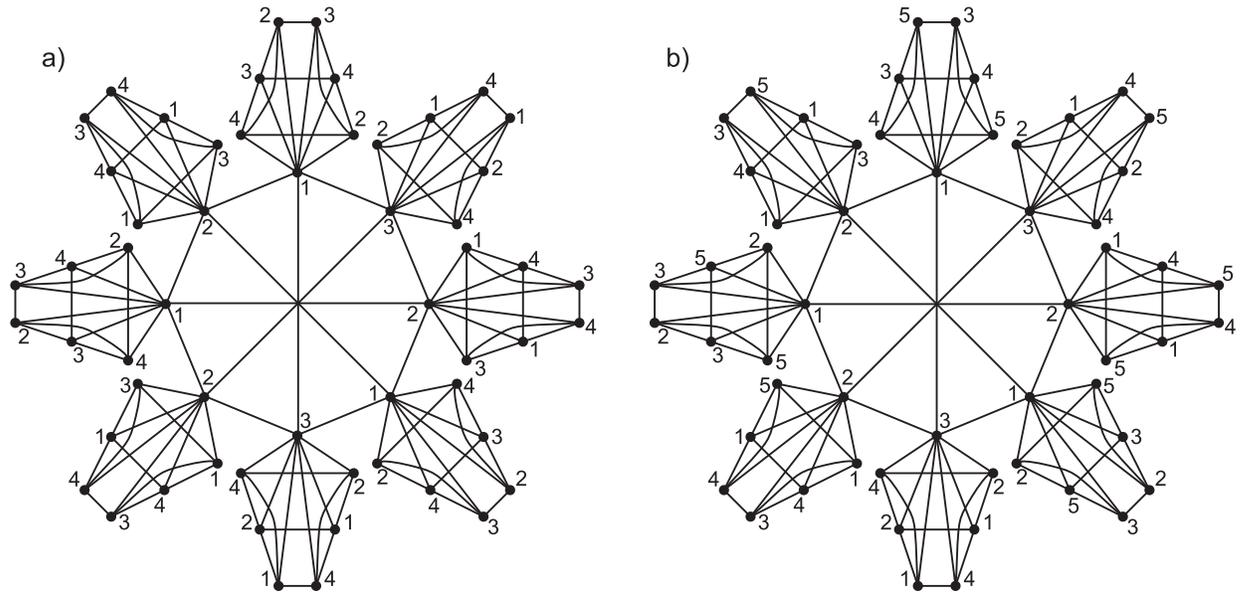} 
\caption{An example of coloring of $W \circ P$, where $W$ is the Wagner graph and $P$ is the prism graph: a) ordinary 4-coloring; b) equitable 5-coloring.}
\label{zPrism}
\end{center}
\end{figure}

Since $|V(G \circ H)|=(m+1)n$, the color cardinality sequence $\mathcal{C}=(c_1,c_2,c_3,c_4)$ of the above 4-coloring of $G\circ H$ is as follows:
\begin{eqnarray*}
\Big(& &\left\lceil n/3 \right\rceil + \left\lceil (n-1)/3\right\rceil \left\lceil (m-2)/3\right\rceil + \left\lceil (n-2)/3\right\rceil \left\lceil(m-1)/3\right\rceil, \\
& &\left\lceil n/3\right\rceil\left\lceil m/3\right\rceil + \left\lceil (n-1)/3\right\rceil+\left\lceil(n-2)/3\right\rceil\left\lceil(m-1)/3\right\rceil,\\
& & \left\lceil n/3\right\rceil \left\lceil(m-1)/3\right\rceil +\left\lceil(n-1)/3\right\rceil\left\lceil m/3\right\rceil+\left\lceil(n-2)/3\right\rceil,\\ & & \left\lceil n/3\right\rceil\left\lceil(m-2)/3\right\rceil+\left\lceil(n-1)/3\right\rceil\left
\lceil(m-1)/3\right\rceil+ \left\lceil(n-2)/3\right\rceil\left\lceil m/3\right\rceil\Big),
\end{eqnarray*}
respectively. This 4-coloring is not equitable. We have to recolor some vertices colored with 1, 2, 3 and 4 into 5. 
The number of vertices colored with $i$, $1 \leq i \leq 4$, that have to be recolored is equal to $c_i - \lceil ((m+1)n-i+1)/5\rceil$.

\vspace{3mm}We have the following claims:

\begin{eqnarray}
c_1 - \left\lceil \frac{(m+1)n}{5}\right\rceil & \leq & \left\lfloor \frac{1}{2} \left\lceil\frac{n-2}{3}\right\rceil\right\rfloor \left\lceil\frac{m-1}{3}\right\rceil =  \left\lfloor \frac{1}{2}|C|\right\rfloor|V|,\label{first}\\
c_2 - \left\lceil \frac{(m+1)n-1}{5}\right\rceil &\leq& \left\lfloor \frac{1}{2} \left\lceil\frac{n}{3}\right\rceil\right\rfloor \left\lceil\frac{m}{3}\right\rceil = \left\lfloor \frac{1}{2}|A|\right\rfloor|U|,\\
c_3 - \left\lceil \frac{(m+1)n-2}{5}\right\rceil &\leq& \left\lfloor \frac{3}{4} \left\lceil\frac{n-1}{3}\right\rceil\right\rfloor \left\lceil\frac{m}{3}\right\rceil=\left\lfloor \frac{3}{4}|B|\right\rfloor|U|, \text{ and}\\
c_4 - \left\lceil \frac{(m+1)n-3}{5}\right\rceil &\leq& \left\lceil \frac{1}{2}\left\lceil \frac{n-2}{3} \right\rceil \right\rceil \left\lceil \frac{m-2}{3}\right\rceil+\left \lceil \frac{1}{4} \left\lceil \frac{n-1}{3}\right\rceil \right\rceil \left \lceil \frac{m-1}{3} \right\rceil+\nonumber\\
& &+ \left\lceil \frac{1}{2} \left \lceil \frac{n-2}{3}\right\rceil \right\rceil \left\lceil \frac{m}{3}\right\rceil = \nonumber \\
& &=\left\lceil\frac{1}{2}|A|\right\rceil |W| + \left\lceil\frac{1}{4}|B|\right\rceil |V| + \left\lceil\frac{1}{2}|C|\right\rceil |U|. \label{last}
\end{eqnarray}

\noindent\textit{Proof of inequalities \emph{(\ref{first})-(\ref{last}).}} Let us consider three cases,
$G\in Q_3(s), Q_3(s+1,s,s),$ and $Q_3(s+1,s+1,s)$, and in each case three subcases, $H \in Q_3(t), Q_3(t+1,t,t), Q_3(t+1,t+1,t)$, respectively. The estimation 
technique for the number of vertices that have to be recolored to color 5 is similar to that used in the proof of Theorem \ref{2_3_5}. 

\begin{description}
\item[\textnormal{\emph{Case} 1:}] $G\in Q_3(s)$, \emph{where} $s=2k$ \emph{for some $k\geq 1$.}

\emph{Subcase} 1.1: $H\in Q_3(t)$, \emph{where $t=2l$ for some $l\geq 1$.}

We have $|V(G \circ H)|=(3t+1)3s=5(7kl+k)+kl+k$, while the color cardinality sequence $\mathcal{C}$ of the 4-coloring of $G\circ H$ is  
$\mathcal{C}=(s+2st,s+2st,s+2st,3st)=(8kl+2k,8kl+2k,8kl+2k,12kl)$.

Since in every equitable 5-coloring of $G \circ H$ each of 5 colors has to be used $(7kl+k+\lceil (kl+k)/5 \rceil)$ or $(7kl+k+\lfloor (kl+k)/5 \rfloor)$ times, we have to 
recolor some vertices colored with 1, 2, 3 and 4 into 5. The number of vertices that have to be recolored is as follows:
\begin{itemize}
\item the vertices colored with 1:

$8kl+2k-7kl-k-\lceil (kl+k)/5 \rceil \rceil \leq 2kl=\left\lfloor \frac{1}{2}|C|\right\rfloor|V|$,
\item the vertices colored with 2:

$8kl+2k-7kl-k-\lceil (kl+k-1)/5 \rceil \rceil \leq 2kl=\left\lfloor \frac{1}{2}|A|\right\rfloor|U|$,
\item the vertices colored with 3:

$8kl+2k-7kl-k-\lceil (kl+k-2)/5 \rceil \leq 2kl\leq \left\lfloor \frac{3}{4}|B|\right\rfloor|U|$,

\item the vertices colored with 4:

$12kl-7kl-k-\lceil (kl+k-3)/5 \rceil \leq 4kl+\lceil\frac{k}{2} \cdot 2l\rceil=\\
=\left\lceil\frac{1}{2}|A|\right\rceil |W| + \left\lceil\frac{1}{4}|B|\right\rceil |V| + \left\lceil\frac{1}{2}|C|\right\rceil |U|$.
\end{itemize}

\emph{Subcase} 1.2: $H\in Q_3(t+1,t,t)$, \emph{where $t=2l+1$ for some $l\geq 1$.}

We have $|V(G \circ H)|=(3t+2)3s=5(7kl+6k)+kl$, while the color cardinality sequence $\mathcal{C}$ of the 4-coloring of $G\circ H$ is  $\mathcal{C}=(s+2st,2s+2st,2s+2st,3st+s)=(8kl+6k,8kl+8k,8kl+8k,12kl+8k)$.

Since in every equitable 5-coloring of $G \circ H$ each of 5 colors has to be used $(7kl+6k+\lceil kl/5 \rceil)$ or $(7kl+6k+\lfloor kl/5 \rfloor)$ times, we have to recolor some vertices colored with 1, 2, 3 and 4 into 5. The number of vertices that have to be recolored is as follows:
\begin{itemize}
\item the vertices colored with 1:

$kl-\lceil kl/5 \rceil \leq 2kl+k = \left\lfloor \frac{1}{2}|C|\right\rfloor|V|$,
\item the vertices colored with 2:

$k(l+1)+k-\lceil (kl-1)/5 \rceil \leq 2k(l+1)=\left\lfloor \frac{1}{2}|A|\right\rfloor|U|$,
\item the vertices colored with 3:

$k(l+1)+k-\lceil (kl-2)/5 \rceil \leq \lfloor \frac{3}{4} k \rfloor(2l+2)=\left\lfloor \frac{3}{4}|B|\right\rfloor|U|$,

\item the vertices colored with 4:

$5kl+2k-\lceil (kl-3)/5 \rceil \leq 4kl+2k+\lceil\frac{k}{2}\rceil(2l+1) = \\ =\left\lceil\frac{1}{2}|A|\right\rceil |W| + \left\lceil\frac{1}{4}|B|\right\rceil |V| + \left\lceil\frac{1}{2}|C|\right\rceil |U|$.
\end{itemize}

\emph{Subcase} 1.3: $H\in Q_3(t+1,t+1,t)$, \emph{where $t=2l$ for some $l\geq 1$.}

We have $|V(G \circ H)|=(3t+3)3s=5(7kl+3k)+kl+3k$, while the color cardinality sequence $\mathcal{C}$ of the 4-coloring of $G\circ H$ is $\mathcal{C}=(2s+2st,2s+2st,3s+2st,3st+2s)=(8kl+4k,8kl+4k,8kl+6k,12kl+4k)$.

Since in every equitable 5-coloring of $G \circ H$ each of 5 colors has to be used $(7kl+3k+\lceil (kl+3k)/5 \rceil)$ or $(7kl+3k+\lfloor (kl+3k)/5 \rfloor)$ times, we have to recolor some vertices colored with 1, 2, 3 and 4 into 5. The number of vertices that have to be recolored is as follows:
\begin{itemize}
\item the vertices colored with 1:

$kl+k-\lceil (kl+3k)/5 \rceil \leq 2kl+k=\left\lfloor \frac{1}{2}|C|\right\rfloor|V|$,
\item the vertices colored with 2:

$kl+k-\lceil (kl+3k-1)/5 \rceil \leq 2kl+k=\left\lfloor \frac{1}{2}|A|\right\rfloor|U|$,
\item the vertices colored with 3:

$kl+3k-\lceil (kl+3k-2)/5 \rceil \leq \lfloor\frac{3}{2}k\rfloor(2l+1)=\left\lfloor \frac{3}{4}|B|\right\rfloor|U|$,

\item the vertices colored with 4:

$5kl+k-\lceil (kl+3k-3)/5 \rceil \leq 4kl+k+\lceil\frac{k}{2}\rceil(2l+1)=\\=\left\lceil\frac{1}{2}|A|\right\rceil |W| + \left\lceil\frac{1}{4}|B|\right\rceil |V| + \left\lceil\frac{1}{2}|C|\right\rceil |U|$.
\end{itemize}

\item[\textnormal{\emph{Case} 2:}] $G\in Q_3(s+1,s,s)$, \emph{where $s=2k+1$ for some $k \geq 1$. } The proof follows by a similar argument to that in Case 1, we omit the details.
\item[\textnormal{\emph{Case} 3:}] $G\in Q_3(s+1,s+1,s)$, \emph{where $s=2k$ for some $k\geq 1$.} The proof follows by a similar argument to that in Case 1, we omit the details.
\end{description}

\noindent\textit{End of the proof of inequalities \emph{(\ref{first})-(\ref{last}).}}

\vspace{3mm}
Now, to obtain an equitable 5-coloring of $G \circ H$, we choose the vertices that have to be recolored.
\begin{itemize}
\item Since the number of 1-vertices that have to be recolored to 5 is not greater than $\lfloor\frac{1}{2}|C|\rfloor |V|$, then the vertices colored with 1 are chosen from the partitions $V$ of $\lfloor\frac{1}{2}|C|\rfloor$ copies of $H$ linked to the vertices from partition $C$ of $G$.
\item Similarly, 2-vertices that have to be recolored are chosen from the partitions $U$ of $\lfloor\frac{1}{2}|A|\rfloor$ copies of $H$ linked to the vertices from partition $A$ of $G$.
\item 3-vertices to be recolored are chosen from the partitions $U$ of $\lfloor\frac{3}{4}|B|\rfloor$ copies of $H$ linked to the vertices from partition $B$ of $G$.
\item 4-vertices are chosen from: 
\begin{itemize}
\item partitions $W$ of $\lceil\frac{1}{2}|A|\rceil$ copies of $H$ linked to the vertices from the partition $A$ of $G$ (different copies than in recoloring of 2-vertices),
\item partitions $V$ of $\lceil\frac{1}{4}|B|\rceil$ copies of $H$ linked to the vertices from the partition $B$ of $G$ (different copies than in recoloring of 3-vertices),
\item partitions $U$ of $\lceil\frac{1}{2}|C|\rceil$ copies of $H$ linked to the vertices from the partition $C$ of $G$ (different copies than in recoloring of 1-vertices) (see Fig. \ref{zPrism}b)).
\end{itemize}
\end{itemize}
Taking into account our claim, such recoloring is possible.
\end{proof}

As we have already observed, the lower bound in Theorem \ref{2_3_5} is tight. Also upper bounds in Theorems \ref{2_3_5} and \ref{Cub3_xxx} are tight. There are infinitely 
many coronas $G \circ H$, where $G \in Q_2 \cup Q_3$ and $H\in Q_3$, that require five colors to be equitably colored. For example, in such coronas graph 
$H \in Q_3$ may be built of $3t$ ($t$ must be even) vertices and it must contain $t$ disjoint triangles (cycles $C_3$) (cf. Fig. \ref{grafH}). Let us consider for 
example $G=K_{3,3}$. In the corona $K_{3,3} \circ H$, where $H$ is defined as above, the number of vertices is equal to $36k+6$, for some positive integer $k$. In any 
equitable 4-coloring of the corona, the color sequence must be  $(9k+2,9k+2,9k+1,9k+1)$. Since modifying the tripartite structure of $H$ is impossible (it contains 
$t=2k$ disjoint triangles), such a coloring does not exist for $k\geq 2$. 

\begin{figure}[htb]
\begin{center}
\begin{picture}(100,80)
\put(0,20){\circle*4}
\put(0,40){\circle*4}
\put(0,60){\circle*4}
\put(0,80){\circle*4}

\put(20,0){\circle*4}
\put(40,0){\circle*4}
\put(60,0){\circle*4}
\put(80,0){\circle*4}
\put(100,20){\circle*4}
\put(100,40){\circle*4}
\put(100,60){\circle*4}
\put(100,80){\circle*4}
\thicklines
\put(0,80){\line(1,0){100}}
\thicklines
\put(0,80){\line(1,-1){80}}
\thicklines
\put(100,80){\line(-1,-4){20}}

\thicklines
\put(0,60){\line(1,0){100}}
\thicklines
\put(0,60){\line(1,-1){60}}
\thicklines
\put(100,60){\line(-2,-3){40}}

\thicklines
\put(0,40){\line(1,0){100}}
\thicklines
\put(0,40){\line(1,-1){40}}
\thicklines
\put(100,40){\line(-3,-2){60}}

\thicklines
\put(0,20){\line(1,0){100}}
\thicklines
\put(0,20){\line(1,-1){20}}
\thicklines
\put(100,20){\line(-4,-1){80}}

\thinlines
\put(0,80){\line(5,-1){100}}
\thinlines
\put(0,60){\line(5,-1){100}}
\thinlines
\put(0,40){\line(3,-2){60}}
\thinlines
\put(0,20){\line(2,-1){40}}

\thinlines
\put(100,80){\line(-1,-1){80}}
\thinlines
\put(100,20){\line(-1,-1){20}}
\end{picture}
\caption{An example of graph $H \in Q_3$ for which $\chi_=(G \circ H)=5$, for $G \in Q_3$.}\label{grafH}
\end{center}
\end{figure}
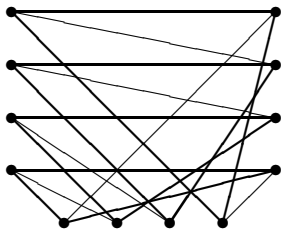

\section{Complexity results}\label{NP}
Although we have only two possible values, 4 and 5, for $\chi_=(G \circ H)$, where $G\in Q_2 \cup Q_3$ and $H \in Q_3$, it is hard to decide which is correct\footnote{graphs considered in this section need not be connected}. All $G, H$ are still cubic.

\noindent We consider the following combinatorial decision problems:
\begin{table}[htb]
\begin{tabular}{ll}
IS$_3(H,k)$: & Given a cubic graph $H$ on $m$ vertices and an integer $k$, the question \\
& is: does $H$ have an independent set $I$ of size at least $k$?\\
\end{tabular}

\vspace{0.5cm}
and its subproblem for $m=10q$, $k=4m/10=4q$, i.e. IS$_3(H,4q)$.
\vspace{0.5cm}
\end{table}

Note that the IS$_3(H,k)$ problem is NP-complete and remains so even if $10|m$ \cite{garey}. This is so because we can enlarge $H$ by adding $j$ $(0\leq j \leq 4)$ 
isolated copies of $K_{3,3}$ to it so that the number of vertices in the new graph is divisible by 10. Graph $H$ has an independent set of size at least $k$ if and only if the new graph 
has an independent set of size at least $k + 3j$.

\begin{lemma}
Problem \emph{IS}$_3(H,4m/10)$ is \emph{NP}-complete.
\end{lemma}
\begin{proof}
Our polynomial reduction is from IS$_3(H,k)$. For an $m$-vertex cubic graph $H$,  $10|m$, and an integer $k$, let $r = |4m/10-k|$. If $k \geq 4m/10$ then we construct a cubic graph 
$G = H + rK_4 + r P$ else we construct $G = H + rK_4 + 2rP + 4rK_{3,3}$, where $P \in Q_3(2)$ is the prism graph. It is easy to see that the answer to problem IS$_3(H,k)$ is 'yes' if and only if the answer to 
problem IS$_3(G,4m/10)$ is 'yes'. 
\end{proof}

\begin{lemma}
Let $H$ be a cubic graph and let $k = 4/10m$, where $m$ is the number of vertices of $H$. The problem of deciding whether $H$ has a coloring of type $(4m/10,3m/10,3m/10)$ is \emph{NP}-complete.
\end{lemma}
\begin{proof}
We prove that $H$ has a coloring of type $(4m/10,3m/10,3m/10)$ if and only if there is an affirmative answer to IS$_3(H,4m/10)$.

Suppose first that $H$ has the above 3-coloring. Then the color class of size $4m/10$ is an independent set that forms a solution to IS$_3(H,4m/10)$.

Now suppose that there is a solution $I$ to IS$_3(H,4m/10)$. Thus $|I| \geq 4m/10$. We know from \cite{bipart} that in this case there exists an independent set $I'$ of size exactly $4m/10$ such 
that the subgraph $H-I'$ is equitably 2-colorable bipartite graph.
This means that $H$ can be 3-colored so that the color sequence is $(4m/10,3m/10,3m/10)$.
\end{proof}

In the following we show that, given such an unequal coloring of $H$, we can color $K_{3,3}\circ H$ equitably with 4 colors.

\begin{enumerate}
\item[(\emph{i})] Color the vertices of $K_{3,3}$ with 4 colors - the color sequence is $(2,2,1,1)$.
\item[(\emph{ii})] Color the vertices in copies of $H=H(U,V,W)$, $|U|=4m/10$, $|V|=|W|=3m/10$, in the following way:
\begin{itemize}
\item vertices in partitions $U$ of $H$ adjacent to a 1-vertex of $K_{3,3}$ are colored with color 2, in partitions $V$ - with 3, and in partitions $W$ - with 4,
\item vertices in partitions $U$ of $H$ adjacent to a 2-vertex of $K_{3,3}$ are colored with color 1, in partitions $V$ - with 3, and in partitions $W$ - with 4,
\item vertices in partition $U$ of $H$ adjacent to the 3-vertex of $K_{3,3}$ are colored with color 1, in partition $V$ - with 2, and in partition $W$ - with 4,
\item vertices in partition $U$ of $H$ adjacent to the 4-vertex of $K_{3,3}$ are colored with color 2, in partition $V$ - with 1, and in partition $W$ - with 3.
\end{itemize}
\end{enumerate} 

Color sequence of the corona is $(15m/10+2,15m/10+2,15m/10+1,15m/10+1)$. 

On the other hand, let us assume that the corona $K_{3,3} \circ H$, where $H \in Q_3(t)$ and $t=10k$, is equitably 4-colorable, where the color sequence for $K_{3,3}$ is $(2,2,1,1)$. 
Since $|V(K_{3,3} \circ H)|=6(3t+1)=18t+6$ and $t=10k$ for some $k$, then each of the four colors in every equitable coloring is used $45k+1$ or $45k+2$ times. Since color 1 
(similarly color 2) can be used only in four copies of $H$, then in at least one copy we have to use it $12k=12t/10$ times. It follows that there must exist an independent set of cardinality $12t/10$ in $H$. Since $H$ has $3t$ vertices, the size of this set is $4m/10$.

The above considerations lead us to the following
\begin{theorem}
The problem of deciding whether $\chi _{=}(K_{3,3} \circ H) =4$ is \emph{NP}-complete even if $H \in Q_3(t)$ and $10 | t$. \hfill $\Box$ \label{twNP}
\end{theorem}

A similar argument implies the following
\begin{corollary}
The problem of deciding whether $\chi _{=}(P \circ H) =4$, where $P$ is the prism graph, is \emph{NP}-complete even if $H \in Q_3(t)$ and $10 | t$. \hfill $\Box$ \label{NP2}
\end{corollary}

In this way we have obtained the full classification of complexity for equitable coloring of cubical coronas.

\section{Conclusion}
In this paper, we presented all the cases of corona of cubic graphs for which 3 colors suffice for equitable coloring. In the remaining cases we have proved constructively that 5 colors 
are enough for equitable coloring. Since there are only two possible values for $\chi_=(G \circ H)$, 
namely 4 or 5, our algorithm is 1-absolute approximate. Due to Theorem \ref{twNP} and Corollary \ref{NP2} the algorithm cannot be improved unless $P=NP$. 
Since time spend to assign a final color to each vertex is constant, the complexity of our algorithm is linear. 
Finally, the algorithm confirms the Equitable Coloring Conjecture \cite{meyer}.

Our results are summarized in Table \ref{tabela}. This table contains also the values of classical chromatic numbers of appropriate coronas and the complexity classification. Let us notice that all cases are polynomially solvable for ordinary coloring.

\begin{table}[htb]
\begin{center}
\begin{tabular}{|c|*{6}{c|}}\hline

\backslashbox[20mm]{$G$}{$H$} & \multicolumn{2}{c|}{$Q_2$} & \multicolumn{2}{c|}{$Q_3$} & \multicolumn{2}{c|}{$Q_4$}\\ \hline
$Q_2, Q_3$ & \emph{3} & {\textbf 3} or {\textbf 4} & \emph{4} &{\textbf 4} or {\textbf 5}$^*$ &\emph{5} &\bf{5} \\ \hline
$Q_4$ & \emph{4} & \bf{4} & \emph{4} & \bf{4} & \emph{5} & \bf{5} \\ \hline
\end{tabular}

\vspace{3mm}
\caption{The exact values of classical chromatic number (in \emph{italics}) and  possible values of the equitable chromatic number (in {\textbf{bold}}) of coronas $G \circ H$. Asterix $(^*)$ means that this case is NP-complete. The other cases are solvable in linear time.}\label{tabela}
\end{center}
\end{table}


\vspace{0.5cm}

\noindent {\large \textbf{Acknowledgments}}

The authors thank Professor Staszek Radziszowski for taking great care in reading our manuscript and making useful suggestions.

\end{document}